\documentclass[12pt,preprint,a4paper,prd,superscriptaddress,nofootinbib,aps]{revtex4-2}
\usepackage[utf8]{inputenc} 

\usepackage{mathtools,amsmath,amssymb,amsthm,physics}

\usepackage[pdftex,dvipsnames,hyperref]{xcolor} 
\usepackage[pdftex,pdfencoding=auto,psdextra,colorlinks=true,allcolors = blue]{hyperref}

\hypersetup{pageanchor=false}

\definecolor{CUDred}{RGB}{255,75,0}  
\definecolor{CUDgreen}{RGB}{3,175,122} 
\definecolor{CUDblue}{RGB}{0,90,255}
\definecolor{SolBase3}{RGB}{253,246,227}
\definecolor{Base00}{RGB}{101,123,131}
\definecolor{Base01}{RGB}{ 88,110,117}
\definecolor{Base02}{RGB}{  7, 54, 66}

\hypersetup{
  colorlinks        = true,
  bookmarksopen     = true,
  bookmarksnumbered = true,
  allcolors         = CUDred
}
\usepackage[]{pagecolor}

\usepackage[T1]{fontenc}
\usepackage{mlmodern}
\renewcommand{\frac}[2]{\genfrac{}{}{0.6pt}{}{#1}{#2}}
\newcommand{\Kozaki}{%
  \author{Hiroshi~\surname{Kozaki}}
  \email{kozaki@ishikawa-nct.ac.jp}
  \affiliation{Department of General Education, National Institute of Technology, Ishikawa College, Ishikawa 929-0392, Japan}
}

\newcommand{\Koike}{%
  \author{Tatsuhiko~Koike} 
  \email{koike@phys.keio.ac.jp}
  \affiliation{Department of Physics, Keio University, Yokohama 223-8522, Japan} 
}

\newcommand{\Ishihara}{%
  \author{Hideki~Ishihara}
  \email{h.ishihara@omu.ac.jp}
  \affiliation{Nambu Yoichiro Institute of Theoretical and Experimental Physics, Osaka Metropolitan University, Osaka 558-8585, Japan}
  \affiliation{Osaka Central Advanced Mathematical Institute, Osaka Metropolitan University, Osaka 558-8585, Japan}
}

\newcommand{\Morisawa}{%
  \author{Yoshiyuki~ Morisawa}
  \email{morisawa@omu.ac.jp}
  \affiliation{Osaka Central Advanced Mathematical Institute, Osaka Metropolitan University, Osaka 558-8585, Japan}
}

\newcommand{\Matsuno}{%
  \author{Satsuki~Matsuno}
  \email{smatsuno43@gmail.com}
  \affiliation{Osaka Central Advanced Mathematical Institute, Osaka Metropolitan University, Osaka 558-8585, Japan}
}

\newcommand{\TensSty}[1]{\boldsymbol{#1}}
\newcommand{\X}{\TensSty{X}}
\newcommand{\Y}{\TensSty{Y}}

\newcommand{\Met}{\boldsymbol{g}}

\newcommand{\PhiTens}{\TensSty{\varphi}}
\newcommand{\XiVect}{\TensSty{\xi}}
\newcommand{\EtaForm}{\TensSty{\eta}}
\newcommand{\ACStruct}{(\PhiTens,\XiVect,\EtaForm)}
\newcommand{\IdTens}{\TensSty{1}}

\newcommand{\CompatMet}{\TensSty{h}}
\newcommand{\ACMStruct}{(\PhiTens,\XiVect,\EtaForm,\CompatMet)}
\newcommand{\PhiForm}{\TensSty{\Phi}}

\newcommand{\OrthoVect}[1]{\TensSty{e}_{#1}}
\newcommand{\DualForm}[1]{\TensSty{\theta}^{#1}}

\newcommand{\NijenTens}[2]{\comm{#1}{#2}}
\newcommand{\LieBrac}[2]{\comm{#1}{#2}}
\newcommand{\LieDer}[1]{\mathcal{L}_{#1}}

\newcommand{\CovDer}[1]{\nabla_{#1}}

\newcommand{\TrCovDerXi}{\tr(\CovDer{} \XiVect)}
\newcommand{\TrPhiCovDerXi}{\tr(\PhiTens \CovDer{} \XiVect)}
\newcommand{\RicTens}{\text{\textbf{Ric}}}
\newcommand{\RicScal}{R}
\newcommand{\EmptSlot}{\;\cdot\;}
\newcommand{\TwoMet}{\TensSty{\gamma}}

\newcommand{\TwoRic}{\RicScal_{2}}
\newcommand{\TmpForm}{\TensSty{\zeta}}

\newcommand{\EinTens}{\TensSty{G}}

\newcommand{\EneMom}{\TensSty{T}}

\DeclareMathOperator{\diag}{diag}
\numberwithin{equation}{section}

\newcommand{\EtaEin}{$\eta$-Einstein}

\newcommand{\Mani}{\mathcal{M}}
\newtheorem{Prop}[]{Proposition}
\newtheorem{Theo}[]{Theorem}
\newcommand{\GaussCurv}{K_2}

\newcommand{\ConstCurvTwoDim}{\varepsilon}
\newcommand{\ConstCurvThreeDim}{k}
\newcommand{\vol}[1]{\textbf{vol}_{#1}}
\newcommand{\GNCI}{\chi}
\newcommand{\GNCII}{\sigma}

\newcommand{\OneForm}{one-form}
\newcommand{\TwoForm}{two-form}

\begin{document}
\title{Geometric deformations of symmetric spacetimes with a string cloud}

\Kozaki

\Matsuno

\Koike

\Morisawa

\Ishihara

\preprint{OCU-PHYS-626}
\preprint{AP-GR-212}
\preprint{NITEP-277}

\begin{abstract}
  We establish a deformation framework for highly symmetric solutions to the Einstein
  equations.
  In this framework, four-dimensional metrics are constructed from three-dimensional
  {\EtaEin} metrics admitting a deformation determined by a single function.
  Under this deformation, the resulting spacetime solves the Einstein equations with a
  string-cloud source. 
  Within this framework, a wide range of symmetric spacetimes can be treated in a unified
  manner. 
  These include FLRW, Kantowski-Sachs, and LRS Bianchi cosmological models
  (including Taub-NUT-(A)dS solutions), 
  as well as Reissner-Nordström-(A)dS black holes admitting spherical, planar, or
  hyperbolic symmetry.
  In the cosmological setting, the deformation leaves the evolution equations for
  the scale factors unchanged, and hence the expansion history coincides with that of the
  corresponding undeformed models. 
  For the deformed Reissner-Nordström-(A)dS black holes, 
  the structure of Killing horizons is insensitive to the deformation. 
\end{abstract}

\maketitle{}

\section{Introduction}
\label{sec:introduction}

Exact solutions of the Einstein equations have played a central role in our understanding
of gravitational phenomena, 
providing concrete models for black holes and cosmological spacetimes
\cite{ExactSolutions,ExactSpacetimes}.  
Well known exact solutions, however, rely heavily on geometric symmetry assumptions and are 
often restricted to vacuum or highly idealized matter models.
While such assumptions greatly simplify the field equations,
they also limit our understanding of how spacetime geometry responds to relaxations of symmetry,
and what types of matter content may be required when deformations are introduced.

From a geometric viewpoint,
it is therefore natural to ask how symmetric spacetimes can be systematically deformed
while remaining solutions to the Einstein equations. 
In particular, one may ask which geometric structures allow such deformations to be
carried out in a controlled manner, and what forms of energy-momentum tensors inevitably 
emerge as a consequence of the deformations. 
Similar questions about the construction of new solutions of the Einstein equations
through deformations have been explored from various perspectives in the literature;
see, e.g., \cite{Vacaru2011DeformSol,Molina2013DeformSol,LasHeras:2019DeformSol}
and references therein. 

Deformations of highly symmetric spacetimes generally break isotropy and therefore require matter 
sources with anisotropic stress to satisfy the Einstein equations. 
From this perspective, phenomenological matter models with anisotropic stress provide a natural and 
effective description of the energy-momentum
tensors required by the deformations.
Among these models, a cloud of strings offers a particularly simple and tractable example,
characterized by stress supported along a single direction \cite{Letelier1979stringcloud}.
Such string cloud models have been studied in a variety of gravitational settings;
see, e.g., \cite{Letelier1983StringCosmology, Nune2025BHSC, Priyokumar2019CosmoSC,
  Singh2020SCAltTheo} and references therein.  
In this work, we employ string clouds as an effective matter model capable of accommodating
the anisotropic stress required by the geometric deformations.

In this paper, we construct a class of spacetime metrics utilizing three-dimensional {\EtaEin}
metrics,
which generalize Einstein metrics by allowing a distinguished direction while
retaining a particularly simple form of the Ricci tensor. 
The resulting spacetime metrics are assumed to accommodate a simple string cloud,
characterized by anisotropic stress supported along a single preferred direction.
From a geometric viewpoint, these metrics can be regarded as deformations of highly symmetric
spacetimes.
Moreover, given a highly symmetric solution of the Einstein equations, possibly including
conventional matter sources such as perfect fluids or electromagnetic fields,
our framework allows one to generate a corresponding solution describing a deformed spacetime,
with the entire effect of the deformation absorbed into the appearance of the string cloud as an 
additional matter component.
Importantly, key features of the original symmetric solution are preserved,
in the sense that certain metric functions, such as scale factors, remain unchanged.

Such a deformation, in which key features of the original symmetric solution are 
preserved while the matter sector is modified,
has been found in the context of black holes supported by string clouds. 
In particular, Boos and Frolov have constructed an exact black hole solution by distorting the 
spherical symmetry of the Reissner-Nordström-(A)dS black hole, effectively replacing the round 
two-sphere with a more general closed two-dimensional geometry, where the deformation is supported 
solely by the introduction of a string cloud \cite{Boos2018BHString}. 
The spacetime construction employed in this paper extends our previous constructions
\cite{Ishihara2022InhomESU,Kozaki2024nullcontact}.
The framework developed in the present work naturally incorporates these special cases and
places them within a unified geometric setting. 
More generally, our framework allows deformations of a broad class of symmetric solutions of
the Einstein equations, including well-known cases such as Friedmann-Lemaître-Robertson-Walker
(FLRW) spacetimes and Taub-NUT-(A)dS spacetimes.

The paper is organized as follows.
In Section~\ref{sec:almost-cont-struct}, we briefly review almost contact metric structures 
and then construct three-dimensional {\EtaEin} metrics. 
In Section~\ref{sec:spac-constr-with},
we use the {\EtaEin} metrics to build spacetimes that admit a simple string cloud and are 
deformations of highly symmetric solutions. 
The energy-momentum tensor of the string cloud is derived from the Nambu-Goto action,
and the Einstein tensors are computed.
Section~\ref{sec:deform-symm-solut} is devoted to the main result,
where we prove that a broad class of highly symmetric spacetimes can be systematically
deformed into exact solutions of the Einstein equations coupled to a string cloud. 
In Section~\ref{sec:examples},
we apply the result to cosmological models and to Reissner-Nordström-(A)dS black hole
spacetimes, 
and discuss the physical implications of the resulting deformed solutions. 
Finally, Section~\ref{sec:conclusion} contains our conclusions.

\section{Almost contact structure and {\EtaEin} metric}
\label{sec:almost-cont-struct}

\subsection{Almost contact structure}
An almost contact (AC) structure is defined on a $(2n+1)$-dimensional manifold.
An AC structure is characterized by a triple $\ACStruct$,
where $\PhiTens$ is a $(1,1)$ tensor field, $\XiVect$ is a vector field, and $\EtaForm$ is
a {\OneForm} satisfying 
\begin{equation}
  \PhiTens^2 = - \IdTens + \XiVect \otimes \EtaForm, \qquad \EtaForm(\XiVect) = 1.  
\end{equation}
It follows that 
\begin{equation}
  \PhiTens \XiVect = 0, \qquad \EtaForm \circ \PhiTens = 0, \qquad \rank\PhiTens = 2n. 
\end{equation}

\subsubsection{Rank}
Let $\dd{\EtaForm}^s$ denote the wedge product of $s$ exterior derivatives of $\EtaForm$, 
\begin{equation}
  \dd{\EtaForm}^s \coloneq \underbrace{\dd{\EtaForm} \wedge \dots \wedge
    \dd{\EtaForm}}_{\text{$s$ factors}}. 
\end{equation}
The rank $r$ of the AC structure $\ACStruct$ is defined as \cite{Olszak1986NACM3D}
\begin{equation}
  r =
  \begin{cases}
    2s     &\qif \dd{\EtaForm}^s \neq 0 \qand \EtaForm \wedge \dd{\EtaForm}^s = 0, \\
    2s + 1 &\qif \EtaForm \wedge \dd{\EtaForm}^s \neq 0  \qand \dd{\EtaForm}^{s+1}= 0. \\
  \end{cases}
\end{equation}

\subsubsection{Normality}
Let $\NijenTens{\PhiTens}{\PhiTens}$ denote the Nijenhuis tensor of $\PhiTens$,
which is defined by 
\begin{equation}
  \NijenTens{\PhiTens}{\PhiTens}(\X, \Y)
  \coloneq
  \PhiTens^2 \LieBrac{\X}{\Y} + \LieBrac{\PhiTens\X}{\PhiTens\Y}
  - \PhiTens \LieBrac{\PhiTens\X}{\Y} - \PhiTens \LieBrac{\X}{\PhiTens \Y}, 
\end{equation}
for any vector fields $\X, \Y$. 
The AC structure is said to be normal if it satisfies%
\footnote{In this paper, the exterior derivative is defined so that 
\begin{equation}
  \dd{\EtaForm}(\X, \Y) = \X \EtaForm(\Y) - \Y \EtaForm(\X) - \EtaForm(\LieBrac{\X}{\Y}).
\end{equation}}
\begin{equation}
  \NijenTens{\PhiTens}{\PhiTens} + \dd{\EtaForm}\otimes \XiVect = 0. 
\end{equation}

A feature of the normality is that $\XiVect$ preserves $\EtaForm$ and $\PhiTens$, i.e., 
$\LieDer{\XiVect} \EtaForm = 0, \quad \LieDer{\XiVect} \PhiTens = 0.$

\subsection{Almost contact metric structure\label{subsec:ACM}}
A Riemannian metric $\CompatMet$ is said to be compatible with the AC structure $\ACStruct$ 
if it satisfies
\begin{equation}
  \CompatMet(\PhiTens\X, \PhiTens\Y) = \CompatMet(\X, \Y) - \EtaForm(\X) \EtaForm(\Y), 
\end{equation}
for any vector fields $\X, \Y$. 
For a compatible metric, the {\OneForm} $\EtaForm$ is the metric dual of the vector field 
$\XiVect$, i.e., 
\begin{equation}
  \label{eq:MetDualEtaXi}
  \CompatMet(\XiVect, \EmptSlot) = \EtaForm. 
\end{equation}
An AC structure $\ACStruct$ together with a compatible metric $\CompatMet$ is called the almost
contact metric (ACM) structure, and is denoted by $\ACMStruct$. 

In $(2n+1)$ dimensions,
one can construct a local orthonormal frame
$\qty{\OrthoVect{1},\OrthoVect{2},\dots,\OrthoVect{2n+1}}$ such that  
\begin{equation}
  \OrthoVect{1} = \XiVect,
  \quad 
  \PhiTens \OrthoVect{2k} = - \OrthoVect{2k+1},
  \quad \PhiTens \OrthoVect{2k+1} = \OrthoVect{2k}, \quad (k = 1, \dots, n). 
\end{equation}
This frame is referred to as the $\varphi$-basis. 
We take the coframe $\qty{\DualForm{1}, \DualForm{2},\dots, \DualForm{2n+1}}$
with $\DualForm{1} = \EtaForm. $ 
It is shown that, in three dimensions, any orthonormal frame with $\OrthoVect{1} =
\XiVect$ is a $\varphi$-basis up to orientation, that is, up to changing $\OrthoVect{2}$
and $\OrthoVect{3}$.

From an ACM structure $\ACMStruct$, a {\TwoForm} $\PhiForm$ is defined by 
\begin{equation}
  \PhiForm(\X, \Y) \coloneq \CompatMet(\X, \PhiTens \Y),
\end{equation}
for any vector fields $\X, \Y$. 
Using this {\TwoForm}, the volume form $\vol{\CompatMet}$ is given by
\begin{equation}
  \label{eq:VolForm}
  \vol{\CompatMet} = \frac{1}{n!} \EtaForm \wedge \PhiForm^n. 
\end{equation}
Indeed, in a $\varphi$-basis, $\PhiForm$ is written as
\begin{equation}
  \PhiForm = \DualForm{2} \wedge \DualForm{3}
  + \dots + \DualForm{2n} \wedge \DualForm{2n+1}. 
\end{equation}

The {\TwoForm} $\PhiForm$ is used to define distinguished ACM structures,
such as contact metric structures, Sasakian structures, and quasi-Sasakian structures. 

\subsubsection{Contact metric structure}

An ACM structure $\ACMStruct$ is called a contact metric structure if it holds that
\begin{equation}
  \PhiForm = \frac{1}{2} \dd{\EtaForm}. 
\end{equation}
A contact metric structure such that $\XiVect$ is a Killing vector field is called a K-contact
structure.

\subsubsection{Sasakian structure}

A normal ACM structure $\ACMStruct$ is called Sasakian if it is contact. 
From Eq.~\eqref{eq:VolForm}, we have
$\EtaForm \wedge \dd{\EtaForm}^{n} = n! \, \vol{\CompatMet} \neq 0$. 
Thus, the Sasakian structure has the maximum rank. 

The Sasakian structure ensures various curvature properties.
A notable one is that, in three dimensions, the Ricci tensor takes a simple form as
\begin{equation}
  \RicTens = f_1 \CompatMet + f_2 \EtaForm \otimes \EtaForm,
\end{equation}
where $f_1, f_2$ are functions. 
In this paper, we will refer to the metrics that yield Ricci tensors of this form as
{\EtaEin}. 

\subsubsection{Quasi-Sasakian structure}

A normal ACM structure is called a quasi-Sasakian structure if it satisfies 
\begin{equation}
  \dd{\PhiForm} = 0.
\end{equation}
The rank of the quasi-Sasakian structure is not necessarily maximum. 
A remarkable property of the quasi-Sasakian structure is that $\XiVect$ is a Killing
vector field \cite{Blair1967QuasiSasaki}. 

\subsection{Three-dimensional {\EtaEin} metrics on normal ACM structures}
In this subsection, we review some results on three-dimensional normal ACM (NACM)
structures from \cite{Olszak1986NACM3D} and write down {\EtaEin} metrics explicitly. 

For a NACM structure in three dimensions, the following equation holds
\begin{equation}
  \CovDer{\X} \XiVect = - \frac{\TrCovDerXi}{2} \PhiTens^2 \X - \frac{\TrPhiCovDerXi}{2} \PhiTens \X, 
\end{equation}
where $\CovDer{}$ denotes the Levi-Civita connection, and
$\TrCovDerXi$ and $\TrPhiCovDerXi$ are the traces of the maps $\X \mapsto \CovDer{\X} \XiVect$ and 
$\X \mapsto \PhiTens \CovDer{\X} \XiVect$.
In terms of a $\varphi$-basis,
we have
\begin{equation}
  \CovDer{i} \xi_j = \CompatMet(\CovDer{\OrthoVect{i}}\XiVect, \OrthoVect{j})
  =
  \frac{1}{2}
  \mqty[
  0 & 0 & 0 \\
  0 &   \TrCovDerXi & \TrPhiCovDerXi \\
  0 & - \TrPhiCovDerXi & \TrCovDerXi 
  ].
\end{equation}
This implies that $\XiVect$ is a shear-free geodesic vector field,
with $\TrCovDerXi$ and $\TrPhiCovDerXi$ determining the expansion and the twist,
respectively.
Moreover, $\XiVect$ satisfies the Killing equation
$\CovDer{i}\xi_j + \CovDer{j}\xi_i = 0$ if and only if $\TrCovDerXi$ vanishes.

One also has the relations 
\begin{align}
  \label{eq:ExtDerPhiForm}
  \dd{\PhiForm} &= \TrCovDerXi \, \EtaForm \wedge \PhiForm, \\
  \label{eq:ExtDerEtaForm}
  \dd{\EtaForm} &= \TrPhiCovDerXi \PhiForm.
\end{align}
These equations imply that the rank of a NACM structure is either one or three. 
Furthermore, we obtain the following proposition.
\begin{Prop}\label{prop}
  In the case where the rank of a NACM structure is three,
  the following three conditions are mutually equivalent,
  whereas in the case where the rank is one, only conditions (\ref{enu:KV}) and
  (\ref{enu:QSasaki}) are mutually equivalent. 
  \begin{enumerate}
    \renewcommand{\labelenumi}{(\arabic{enumi})}
  \item $\XiVect$ is a Killing vector field, i.e., $\TrCovDerXi = 0$
    \label{enu:KV}
  \item NACM is quasi Sasakian, i.e., $\dd{\PhiForm = 0}$
    \label{enu:QSasaki}
  \item $\XiVect \TrPhiCovDerXi = 0$
    \label{enu:LieDerTwist}
  \end{enumerate}  
\end{Prop}
\begin{proof}
  The implication (\ref{enu:KV}) $\Leftrightarrow$ (\ref{enu:QSasaki}) is immediate from
  Eq.~\eqref{eq:ExtDerPhiForm}.

  In the case where the rank is three,
  the implications (\ref{enu:QSasaki}) $\Rightarrow$ (\ref{enu:LieDerTwist})
  and (\ref{enu:LieDerTwist}) $\Rightarrow$ (\ref{enu:KV}) follow readily from the formula, 
  \begin{equation}
    \LieDer{\XiVect} \TrPhiCovDerXi = - \TrPhiCovDerXi \TrCovDerXi. 
  \end{equation}
  This is obtained by taking the exterior derivative of Eq.~\eqref{eq:ExtDerEtaForm}
  and contracting with $\XiVect$. 
\end{proof}

Curvature properties are clarified as well.
In particular, the Ricci tensor is given by 
\begin{multline}
  \label{eq:RicTens}
  \RicTens(\X,\Y)
  =
  \qty[
  \frac{\RicScal}{2} + \frac{\XiVect \TrCovDerXi}{2}
  + \frac{\TrCovDerXi^2}{4} - \frac{\TrPhiCovDerXi^2}{4}
  ] \CompatMet(\X, \Y)
  \\
  - \qty[
  \frac{\RicScal}{2} + \frac{\XiVect \TrCovDerXi}{2} + \frac{3 \TrCovDerXi^2}{4} -
  \frac{3\TrPhiCovDerXi^2}{4}
  ] \EtaForm(\X) \EtaForm(\Y)
  \\
  - \EtaForm(\X) \TmpForm(\Y) - \TmpForm(\X) \EtaForm(\Y), 
\end{multline}
where $\RicScal$ is the Ricci scalar, and $\TmpForm$ is a {\OneForm} such that 
\begin{equation}
  \TmpForm(\X)
  = 
  \frac{\X \TrCovDerXi}{2} + \frac{(\PhiTens \X) \TrPhiCovDerXi}{2}.  
\end{equation}

\newcommand{\TwoDimFlatMet}{\dd{x}^2+\dd{y}^2}
Now we seek {\EtaEin} metrics on NACM structures $\ACMStruct$.

First, we take local coordinates $w,x,y$ such that 
\begin{align}
  \XiVect    &= \pdv{w}, \\
  \EtaForm   &= \dd{w} + f \dd{y}, \\
  \label{eq:NormACMMet}
  \CompatMet &= \qty(\dd{w} + f \dd{y})^2 + e^{2\Omega}\qty(\TwoDimFlatMet),  
\end{align}
where $f$ is a function of $x,y$ with $f_{,x}\neq 0$ in the rank-three case, and
$f = 0$ in the rank-one case, and $\Omega$ is a function of $w,x,y$.
The existence of such coordinates is established in Appendix \ref{app:MetInNormACM},
where a different normalization of $\Omega$ is used. 
In this setting, $\TrCovDerXi$ is given by 
\begin{equation}
  \TrCovDerXi = 2\Omega_{,w}. 
\end{equation}

Next, we use the orthonormal frame defined by 
\begin{equation}
  \OrthoVect{1} = \XiVect = \pdv{w},
  \quad \OrthoVect{2} = e^{-\Omega} \pdv{x},
  \quad \OrthoVect{3} = e^{-\Omega} \qty(\pdv{y} - f \pdv{w}). 
\end{equation}
As noted in Sec.~\ref{subsec:ACM}, this frame is a $\varphi$-basis up to orientation,
that is, either $\qty{\OrthoVect{1}, \OrthoVect{2},\OrthoVect{3}}$ or
$\qty{\OrthoVect{1},\OrthoVect{3},\OrthoVect{2}}$ is a $\varphi$-basis. 
We assume here that $\qty{\OrthoVect{1}, \OrthoVect{2},\OrthoVect{3}}$ is a
$\varphi$-basis.
Then, from Eq.~\eqref{eq:ExtDerEtaForm},  we have
\begin{equation}
  \TrPhiCovDerXi = \dd{\EtaForm}(\OrthoVect{2},\OrthoVect{3}) = e^{-2\Omega} f_{,x}.  
\end{equation}
From Eq.~\eqref{eq:RicTens},
a necessary and sufficient condition for the metric $\CompatMet$
to be {\EtaEin} is given by 
\begin{equation}
  \TmpForm(\OrthoVect{2}) = \TmpForm(\OrthoVect{3}) = 0. 
\end{equation}
This reads
\begin{align}
  \label{eq:CondEtaEinRank3A}
  \pdv{x} 2 \Omega_{,w} - \qty(\pdv{y} - f \pdv{w}) \qty(e^{-2\Omega}f_{,x}) &= 0,
  \\
  \label{eq:CondEtaEinRank3B}
  \qty(\pdv{y} - f \pdv{w}) 2 \Omega_{,w} + \pdv{x} \qty(e^{-2\Omega}f_{,x}) &= 0. 
\end{align}
We consider these equations in the rank-one and rank-three cases. 

\subsubsection{Rank-one {\EtaEin{}} metrics}

In the rank-one case, the function $f$ of the compatible metric \eqref{eq:NormACMMet} vanishes, 
and hence, Eqs.~\eqref{eq:CondEtaEinRank3A} and \eqref{eq:CondEtaEinRank3B} read
\begin{equation}
  \pdv{x} \Omega_{,w} = \pdv{y} \Omega_{,w} = 0. 
\end{equation}
Solving this equation, we see that the {\EtaEin} metric is given as a warped product
\newcommand{\WFRI}{\alpha}
\begin{equation}
  \label{eq:MetEtaEinRankOne}  
  \CompatMet = \dd{w}^2 + {\WFRI}^2(w) \TwoMet, 
\end{equation}
where $\TwoMet$ is a two-dimensional metric. 

For convenience,
we use Gaussian normal coordinates $\GNCI,\GNCII$ in two dimensions such that $\TwoMet$ is
written as 
\begin{equation}
  \TwoMet = \dd{\GNCI}^2 + \gamma^2(\GNCI,\GNCII) \dd{\GNCII}^2. 
\end{equation}
We note that this {\EtaEin} metric is not quasi-Sasakian metric unless the function
$\WFRI(w)$ is constant.

\subsubsection{Rank-three {\EtaEin{}} metrics}

In contrast to the rank-one case, we do not pursue a general solution of
Eqs.~\eqref{eq:CondEtaEinRank3A} and \eqref{eq:CondEtaEinRank3B}. 
Instead, we restrict to the quasi-Sasakian case,
which yields simple rank-three {\EtaEin} metrics that can be used to construct spacetimes
with a string cloud.  
Solutions of Eqs.~\eqref{eq:CondEtaEinRank3A} and \eqref{eq:CondEtaEinRank3B} under weaker
assumptions are discussed in Appendix~\ref{sec:solutions-eqs}.

In the quasi-Sasakian case,
it follows from Proposition~\ref{prop} that $\Omega_{,w} = 0$, and hence, 
Eqs.~\eqref{eq:CondEtaEinRank3A} and \eqref{eq:CondEtaEinRank3B} read 
\begin{equation}
  \pdv{x} \qty(e^{-2\Omega} f_{,x}) = \pdv{y} \qty(e^{-2\Omega} f_{,x}) = 0. 
\end{equation}
This is readily solved as $e^{-2\Omega}f_{,x} = \text{const.}$ 
For later convenience, we will work in Gaussian normal coordinates so that the quasi-Sasakian 
metric is written as 
\begin{equation}
  \CompatMet
  =
  (\dd{w} + f(\GNCI,\GNCII) \dd{\GNCII})^2
  +
  \dd{\GNCI}^2 + \gamma^2(\GNCI,\GNCII)\dd{\GNCII}^2.
\end{equation}
This coordinate choice is justified in Appendix~\ref{app:MetInNormACM}.
For this metric form, the metric becomes {\EtaEin} if and only if 
\begin{equation}
  f(\GNCI,\GNCII) = C \int^{\GNCI} \gamma(\GNCI',\GNCII) \dd{\GNCI'}, 
\end{equation}
where $C$ is a non-vanishing constant.
In particular, when $C = 2$, the quasi-Sasakian structure reduces to a Sasakian structure.
Scaling the coordinates, 
we can write the metric as a constant multiple of the Sasakian metric 
\begin{equation}
  \label{eq:MetEtaEinRankThree}
  \CompatMet
  =
  \frac{1}{C^2} \qty[(\dd{w} + f(\GNCI,\GNCII) \dd{\GNCII})^2 + \dd{\GNCI}^2 +
  \gamma^2(\GNCI,\GNCII) \dd{\GNCII}^2], 
  \quad
  f(\GNCI,\GNCII) = \int^{\GNCI} \gamma(\GNCI',\GNCII) \dd{\GNCI'}.  
\end{equation}
In the following, we shall set the overall factor $C=1$, 
because this factor can be absorbed into the scale factors
of the spacetime metric constructed in the next section.

\section{Spacetime construction with a string cloud}
\label{sec:spac-constr-with}

We construct spacetimes $(\Mani,\Met)$ admitting simple string clouds by using
three-dimensional NACM structures $\ACMStruct$.
The metrics $\CompatMet$ are assumed to be {\EtaEin}
and are given by Eq.~\eqref{eq:MetEtaEinRankOne} in the rank-one case and 
by Eq.~\eqref{eq:MetEtaEinRankThree} with $C=1$ in the rank-three case.

\subsection{Metrics}
\label{subsec:SCMetric}
We start from a static spacetime with the metric,
\begin{equation}
  \label{eq:StaticMet}
  \tilde{\Met} = - \dd{t}^2 + \CompatMet, 
\end{equation}
which has been studied in Ref.~\cite{Ishihara2022InhomESU}.
This spacetime admits a string worldsheet tangent to the timelike Killing vector field
$\pdv*{t}$ and $\XiVect$.
For such a cohomogeneity-one configuration, 
the Nambu-Goto equation reduces to the geodesic equation on the orbit space $\Mani/G$,
where $G$ is the isometry group generated by $\pdv*{t}$ \cite{Ishihara2005C1String}. 
In the present setup, the orbit space is identified with the NACM manifold, 
where $\XiVect$ is geodesic; hence the Nambu-Goto equation is satisfied.

For the special case in which $\XiVect$ is a Killing vector field
(equivalently, the NACM structure is quasi-Sasakian), 
the static metric \eqref{eq:StaticMet} admits an extension. 
In this case, the string worldsheet is cohomogeneity-one with respect to $\XiVect$, 
and the Nambu-Goto equation reduces to the geodesic equation on the orbit
space modulo the isometry group generated by $\XiVect$.
Extending the original orbit space metric, 
\newcommand{\ThreeLorentzMet}{\TensSty{h}'}
\newcommand{\WFLorentz}{\beta}
we introduce a time-dependent scale factor $\WFLorentz(t)$, 
\begin{equation}
  \ThreeLorentzMet
  =
  -\dd{t}^2 + \WFLorentz^2(t) \qty[ \dd{\GNCI}^2 + \gamma^2(\GNCI,\GNCII) \dd{\GNCII}^2], 
\end{equation}
in which $\pdv*{t}$ remains geodesic and the Nambu-Goto equation continues to hold. 
The corresponding spacetime metric is written as 
\begin{equation}
  \label{eq:DynaicalMet}
  \tilde{\Met} = - \dd{t}^2 + (\dd{w}^2 + f(\GNCI,\GNCII) \dd{\GNCII})^2 + \WFLorentz^2(t)
  \qty[\dd{\GNCI}^2 +  \gamma^2(\GNCI,\GNCII) \dd{\GNCII}^2],
\end{equation}
where
\begin{equation}
  f(\GNCI,\GNCII) =
  \begin{cases}
    0, & \text{rank $1$}, \\
    \displaystyle
    \int^\GNCI \gamma(\GNCI',\GNCII) \dd{\GNCI'}, & \text{rank $3$}.
  \end{cases}
\end{equation}
This metric is no longer static but retains spatial symmetry associated with $\XiVect
= \pdv*{w}$.

We now perform a conformal rescaling of the metrics $\tilde{\Met}$ given in
Eqs.~\eqref{eq:StaticMet} and \eqref{eq:DynaicalMet} by setting 
\begin{equation}
  \label{eq:ConfScal}
  \Met = a^{2} \tilde{\Met}. 
\end{equation}
For these metrics, we require that the two-dimensional timelike surfaces spanned by $\pdv*{t}$ and 
$\XiVect$ satisfy the Nambu-Goto equation. 
This requirement fixes the conformal factor $a$ to be 
\begin{equation}
  a =
  \begin{cases}
    a(t,w), & \text{rank $1$}, \\
    a(t),   & \text{rank $3$}. 
  \end{cases}
\end{equation}
Accordingly, the spacetime metrics take the following forms: 
\begin{enumerate}
\item In the rank-one case,
  \begin{equation}
    \label{eq:SCMetR1Uni}
    \Met
    =
    \Met[a,b,\gamma]
    =
    a^2(t,w) \qty(-\dd{t}^2+\dd{w}^2)
    +
    b^2(t,w) \qty(\dd{\GNCI}^2+\gamma^2(\GNCI,\GNCII)\dd{\GNCII}^2),
  \end{equation}
  where
  \begin{equation}
    b(t,w) = 
    \begin{cases}
      a(t,w) \WFRI(w), \\
      a(t,w) \WFLorentz(t). 
    \end{cases}
  \end{equation}
\item In the rank-three case, 
  \begin{equation}
    \label{eq:SCMetSPR3}
    \Met
    =
    \Met[a,b,\gamma]
    =
    - \dd{t}^2
    +
    a^2(t) (\dd{w} + f(\GNCI,\GNCII) \dd{\GNCII})^2
    +
    b^2(t) \qty(\dd{\GNCI}^2 + \gamma^2(\GNCI,\GNCII)\dd{\GNCII}^2), 
  \end{equation}
  where
  \begin{equation}
    \label{eq:DeterminF}
    f(\GNCI,\GNCII) = \int^{\GNCI} \gamma(\GNCI',\GNCII) \dd{\GNCI'}. 
  \end{equation}
  We note that the time coordinate has been changed and the scale factors $a(t)$ 
  and $b(t)$ are redefined. 
\end{enumerate}
This completes the construction of the spacetimes $(\Mani,\Met)$ admitting a simple string
cloud.

For a generic choice of the function $\gamma$, the metrics $\Met[a,b,\gamma]$ exhibit no manifest
Killing symmetries, apart from the Killing vector field $\pdv*{w}$ in the rank-three case. 
In contrast, for special choices of $\gamma$ such that the two-dimensional metric
$\dd{\GNCI}^2 + \gamma^2(\GNCI,\GNCII) \dd{\GNCII}^2$ has a constant Gaussian curvature
$\ConstCurvTwoDim = -1,0,1$, the metrics $\Met[a,b,\gamma]$ 
acquire additional Killing symmetries. 
In this case one may take 
\begin{equation}
  \gamma(\GNCI,\GNCII) = \Sigma(\GNCI;\ConstCurvTwoDim),
\end{equation}
where 
\begin{equation}
  \quad 
  \Sigma(\GNCI;\ConstCurvTwoDim) \coloneq
  \begin{cases}
    \sin \GNCI,   & \ConstCurvTwoDim = 1, \\
    \GNCI,        & \ConstCurvTwoDim = 0, \\
    \sinh \GNCI, & \ConstCurvTwoDim = -1, 
  \end{cases}
\end{equation}
In fact, in the rank-one case, $\Met[a,b,\Sigma]$ admits spherical ($\ConstCurvTwoDim=1$), planar
($\ConstCurvTwoDim=0$), or hyperbolic ($\ConstCurvTwoDim=-1$) symmetry,
while in the rank-three case, $\Met[a,b,\Sigma]$ admits the symmetry of the locally rotationally
symmetric (LRS) Bianchi type II ($\ConstCurvTwoDim=0$), VIII ($\ConstCurvTwoDim=-1$), or IX
($\ConstCurvTwoDim=1$). 
From this perspective, we regard $\Met[a,b,\gamma]$ as deformations of the
metric $\Met[a,b,\Sigma]$ admitting the above Killing symmetries. 
In what follows, we refer to the metrics $\Met[a,b,\gamma]$ as deformed metrics. 
The deformation is controlled by the function $\gamma(\GNCI,\GNCII)$.

\subsection{String cloud}

In the spacetimes with the deformed metrics $\Met[a,b,\gamma]$, 
the string worldsheets spanned by $\pdv*{t}$ and $\XiVect=\pdv*{w}$ satisfy the Nambu-Goto
equation.
Such strings extend along the spacelike vector field $\XiVect$. 
We refer to these strings as $\xi$-strings. 
A string cloud is constructed as a family of $\xi$-strings. 

\subsubsection{Rank-one case}

The deformed metric $\Met[a,b,\gamma]$ is given by Eq.~\eqref{eq:SCMetR1Uni}.
We use the orthonormal frame $\qty{\OrthoVect{a}}$ given as follows:
\begin{equation}
  \label{eq:OrthoNorFramR1}
  \OrthoVect{0} = \frac{1}{a} \pdv{t},
  \quad \OrthoVect{1} = \frac{1}{a} \pdv{w},
  \quad \OrthoVect{2} = \frac{1}{b} \pdv{\GNCI},
  \quad \OrthoVect{3} = \frac{1}{b \gamma} \pdv{\GNCII}. 
\end{equation}

First, we consider the energy-momentum tensor of a single $\xi$-string.
Since the worldsheet is spanned by $\pdv*{t}$ and $\XiVect=\pdv*{w}$,
the coordinates $\GNCI, \GNCII$ are constant along the worldsheet.
Let $\GNCI_0, \GNCII_0$ be the constants.
The frame components of the energy-momentum tensor $\EneMom$ are calculated as
\begin{equation}
  \label{eq:EneMomSingleString}
  T_{ab}
  =
  \EneMom(\OrthoVect{a},\OrthoVect{b})
  =
  \mu \frac{\delta(\GNCI-\GNCI_0)\delta(\GNCII-\GNCII_0)}{b^2(t,w) \sqrt{\gamma^2(\GNCI,\GNCII)}}
  \diag(1,-1,0,0), 
\end{equation}
where
$\mu$ is the line density of the string energy.
The derivation is given in Appendix~\ref{appendix:deriv}. 
\newcommand{\TwoMani}{\mathcal{M}_{2}}
We introduce a two-dimensional string source space $\TwoMani$, endowed with the metric 
$\dd{\GNCI}^2+\gamma^2\dd{\GNCII}^2$,
so that $\delta(\GNCI - \GNCI_0)\delta(\GNCII - \GNCII_0)/\sqrt{\gamma^2}$
is interpreted as an invariant delta function.

Next, we extend this energy-momentum tensor to a discrete collection of $\xi$-strings. 
Let $N$ be the number of $\xi$-strings with (possibly different) line energy densities
$\mu_i, \ i = 1, \dots, N$.  
Let the worldsheets be given by $\GNCI = \GNCI_i, \GNCII = \GNCII_i$. 
We take the energy-momentum tensor to be given by the sum of the individual contributions,
\begin{equation}
  T_{ab}
  =
  \frac{1}{b^2(t,w)}
  \qty[
  \sum_i
  \mu_i\frac{\delta(\GNCI-\GNCI_i)\delta(\GNCII-\GNCII_i)}{\sqrt{\gamma^2(\GNCI,\GNCII)}}
  ]
  \diag(1,-1,0,0). 
\end{equation}
The factor $\sum_i \mu_i \delta(\GNCI-\GNCI_i) \delta(\GNCII-\GNCII_i) / \sqrt{\gamma^2}$ is interpreted as
the strings' linear energy density per unit area on the string source space $\TwoMani$.

Finally, we construct a string cloud as a continuous distribution of the $\xi$-strings.
In the continuum limit, 
the energy-momentum tensor is obtained by replacing
the above discrete density with a non-negative continuous areal density
$\mathcal{E}(\GNCI,\GNCII)$, namely, 
\begin{equation}
  \label{eq:EneMomSCR1}
  T_{ab}
  =
  \frac{\mathcal{E}(\GNCI,\GNCII)}{b^2(t,w)}
  \diag(1,-1,0,0)
\end{equation}
This energy-momentum tensor satisfies the local conservation law
$\CovDer{} \cdot\EneMom = 0$, 
supporting the validity of the construction. 
The energy-momentum tensor also satisfies all the standard energy conditions,
namely the null, weak, dominant, and strong energy conditions.
A key feature of the string cloud construction is that it admits a freedom in the
choice of the areal density function $\mathcal{E}(\GNCI,\GNCII)$ on the string source space
$\TwoMani$.  

\subsubsection{Rank-three case}

The deformed metric $\Met[a,b,\gamma]$ is given by Eq.~\eqref{eq:SCMetSPR3}. 
We use the orthonormal frame
\begin{equation}
  \label{eq:OrthoNorFramR3}
  \OrthoVect{0} = \pdv{t},
  \quad \OrthoVect{1} = \frac{1}{a} \pdv{w},
  \quad \OrthoVect{2} = \frac{1}{b} \pdv{\GNCI},
  \quad \OrthoVect{3} = \frac{1}{b \gamma} \qty(\pdv{\GNCII} - f \pdv{w}). 
\end{equation}
The energy-momentum tensor can be constructed in the same manner as in the rank-one case.
The resulting expression has the same functional form as Eq.~\eqref{eq:EneMomSCR1},
with the difference that the function $b$ now depends only on $t$, namely, 
\begin{equation}
  \label{eq:EneMomSCR3}
  T_{ab}
  =
  \frac{\mathcal{E}(\GNCI,\GNCII)}{b^2(t)}
  \diag(1,-1,0,0). 
\end{equation}
This energy-momentum tensor also satisfies the local conservation law 
and all the standard energy conditions,
and admits a freedom in the choice of the areal 
density function $\mathcal{E}(\GNCI,\GNCII)$ on $\TwoMani$. 

\subsection{The Einstein tensor}
We compute the Einstein tensors $\EinTens[a,b,\gamma]$ of the deformed metrics $\Met[a,b,\gamma]$,
with particular attention to the $\gamma$-dependent contributions. 

\subsubsection{Rank-one case}

The metric and the orthonormal frame are given in Eqs.~\eqref{eq:SCMetR1Uni} and
\eqref{eq:OrthoNorFramR1}. 
The nonvanishing components are calculated as
\begin{align}
  \label{eq:EinTensR100}
  G_{00}[a,b,\gamma]
  &= \frac{1}{a^2} \qty[
    2 \frac{a'}{a} \frac{b'}{b}
    - \qty(\frac{b'}{b})^2
    - 2 \frac{b''}{b}
    + 2 \frac{\dot{a}}{a} \frac{\dot{b}}{b}
    + \qty(\frac{\dot{b}}{b})^2
    ]
    +
    \frac{\GaussCurv[\gamma]}{b^2}, \\
  \label{eq:EinTensR101}
  G_{01}[a,b,\gamma]
  &= - \frac{2}{a^2}
    \qty[
    - \frac{\dot{a}}{a} \frac{b'}{b}
    - \frac{a'}{a} \frac{\dot{b}}{b}
    + \frac{\dot{b}'}{b}
    ], \\
  \label{eq:EinTensR111}
  G_{11}[a,b,\gamma]
  &= \frac{1}{a^2} \qty[
    2 \frac{a'}{a} \frac{b'}{b}
    + \qty(\frac{b'}{b})^2
    + 2 \frac{\dot{a}}{a} \frac{\dot{b}}{b}
    - \qty(\frac{\dot{b}}{b})^2
    - 2 \frac{\ddot{b}}{b}
    ]
    -
    \frac{\GaussCurv[\gamma]}{b^2}, \\
  \label{eq:EinTensR122}
  G_{22}[a,b,\gamma]
  &= G_{33}[a,b,\gamma]
    =\frac{1}{a^2} \qty[
    - \qty(\frac{a'}{a})^2
    + \frac{a''}{a}
    + \frac{b''}{b}
    + \qty(\frac{\dot{a}}{a})^2
    - \frac{\ddot{a}}{a}
    - \frac{\ddot{b}}{b}
    ], 
\end{align}
where the dots and the primes denote the differentiation with respect to $t$ and $w$, 
and $\GaussCurv[\gamma]$ is the Gaussian curvature of the string source space $\TwoMani$,
which is explicitly given by
\begin{equation}
  \GaussCurv[\gamma] = -\frac{\gamma_{,\GNCI\GNCI}}{\gamma}. 
\end{equation}
The $\gamma$-dependent part separates and appears only in $G_{00}$ and $G_{11}$,
in the form $\pm \GaussCurv[\gamma]/b^2$.

\subsubsection{Rank-three case}
The metric and the orthonormal frame are given in Eq.~\eqref{eq:SCMetSPR3} and 
Eq.~\eqref{eq:OrthoNorFramR3}.
We use the same notation as in the rank-one case. 
In contrast, the metric functions $a, b$ are functions of $t$ only. 
The nonvanishing components of the Einstein tensor are calculated as
\begin{align}
  \label{eq:EinTensR300}
  G_{00}[a,b,\gamma]
  &=
    -\frac{a^2}{4b^4} + 2\frac{\dot{a}}{a} \frac{\dot{b}}{b}
    + \qty(\frac{\dot{b}}{b})^2
    + \frac{\GaussCurv[\gamma]}{b^2}, \\
  \label{eq:EinTensR311}
  G_{11}[a,b,\gamma]
  &=
    \frac{3}{4}\frac{a^2}{b^4} - \qty(\frac{\dot{b}}{b})^2 - 2 \frac{\ddot{b}}{b}
    - \frac{\GaussCurv[\gamma]}{b^2}, \\
  \label{eq:EinTensR322}
  G_{22}[a,b,\gamma]
  &= G_{33}[a,b,\gamma]
    =  - \frac{a^2}{4b^4} - \frac{\dot{a}}{a} \frac{\dot{b}}{b}
    - \frac{\ddot{a}}{a} - \frac{\ddot{b}}{b}. 
\end{align}
The $\gamma$-dependent part again appears only in $G_{00}$ and $G_{11}$ and takes the
form $\pm \GaussCurv[\gamma]/b^2$, exactly as in the rank-one case.

\section{Deformations of symmetric solutions}
\label{sec:deform-symm-solut}

We consider the Einstein equations for the deformed metric $\Met[a,b,\gamma]$ given by
Eqs.~\eqref{eq:SCMetR1Uni} and \eqref{eq:SCMetSPR3}.
When $\gamma(\GNCI,\GNCII) = \Sigma(\GNCI;\ConstCurvTwoDim)$,
the metric reduces to $\Met[a,b,\Sigma]$ that admits high Killing symmetries. 
In what follows, we take $\Met[a,b,\Sigma]$ as the starting point and regard the general metric
$\Met[a,b,\gamma]$ as its deformation obtained by replacing $\Sigma$ with $\gamma$. 

We now state our main result.

\begin{Theo}\label{Theo:Defo}
  Let $\Met[a,b,\Sigma]$ be a solution of the Einstein equations
  \begin{equation}
    \EinTens[a,b,\Sigma] + \Lambda \Met[a,b,\Sigma] = 8\pi \EneMom, 
  \end{equation}
  where $\EneMom{}$ is the energy-momentum tensor of matter.
  Then the deformed metric $\Met[a,b,\gamma]$ satisfies the Einstein equations with an additional
  string cloud,
  \begin{equation}
    \label{eq:EinEqSC}
    \EinTens[a,b,\gamma] + \Lambda \Met[a,b,\gamma] =
    8\pi(\EneMom{}+\EneMom_{\text{string cloud}}), 
  \end{equation}
  provided that 
  \begin{equation}
    \label{eq:EinEqString}
    \GaussCurv[\gamma] - \GaussCurv[\Sigma] = 8 \pi \mathcal{E},
  \end{equation}
  where $\GaussCurv$ denotes the Gaussian curvature of the string source space
  $\TwoMani$, 
  and $\mathcal{E}$ is the energy density function of the string cloud on $\TwoMani$. 
\end{Theo}

This Theorem provides a systematic construction of deformed solutions by introducing a string cloud
satisfying Eq.~\eqref{eq:EinEqString}.
In this construction, the problem reduces to solving an equation on the string source space
$\TwoMani$,
with the Gaussian curvature determined by the energy density function. 

\begin{proof}
  From Eqs.~\eqref{eq:EinTensR100}--\eqref{eq:EinTensR122} and 
  \eqref{eq:EinTensR300}--\eqref{eq:EinTensR322}, 
  it follows that under the deformation from $\Sigma$ to $\gamma$, 
  only the following components of the Einstein tensor vary: 
  \begin{align}
    G_{00}[a,b,\gamma]
    &= G_{00}[a,b,\Sigma] + \frac{\GaussCurv[\gamma] - \GaussCurv[\Sigma]}{b^2}, \\
    G_{11}[a,b,\gamma]
    &= G_{11}[a,b,\Sigma] - \frac{\GaussCurv[\gamma] - \GaussCurv[\Sigma]}{b^2}. 
  \end{align}
  The energy-momentum tensor of a string cloud is given by
  Eqs.~\eqref{eq:EneMomSCR1} and \eqref{eq:EneMomSCR3}. 
  Therefore,
  the Einstein equations \eqref{eq:EinEqSC} is satisfied if Eq.~\eqref{eq:EinEqString} holds.
\end{proof}

\section{Examples of the deformations}
\label{sec:examples}

\subsection{The rank-one case}
\label{sec:rank-one-case}

\subsubsection{Deformations of the FLRW universe}
\label{subsec:DefoFLRW}
For the rank-one deformed metric $\Met[a,b,\gamma]$ given by Eq.~\eqref{eq:SCMetR1Uni}, we take
\begin{equation}
  a(t,w) = a(t), \quad b(t,w) = a(t) \Sigma(w;\ConstCurvThreeDim), \quad \ConstCurvThreeDim =
  -1,0,1. 
\end{equation}
Then we have the deformed FLRW metric
\begin{equation}
  \label{eq:MetDefFLRWR1}
  \Met[a,b,\gamma]
  =
  -\dd{t}^2
  + a^2(t)\qty[
  \dd{w}^2
  +
  \Sigma^2(w;\ConstCurvThreeDim) \qty(\dd{\GNCI}^2 + \gamma^2(\GNCI,\GNCII)
  \dd{\GNCII}^2 )
  ]. 
\end{equation}
When $\gamma(\GNCI,\GNCII) = \Sigma(\GNCI;\ConstCurvTwoDim = 1)= \sin \GNCI $, 
the metric reduces to the FLRW metric.
In the FLRW universe model,
the homogeneity and isotropy require the energy-momentum tensor
to take the form of a perfect fluid, 
and the Einstein equations reduces to the Friedmann equation
\begin{equation}
  \qty(\frac{\dot{a}}{a})^2 + \frac{\ConstCurvThreeDim}{a^2 } - \frac{\Lambda}{3} =
  \frac{8\pi}{3} \rho,
\end{equation}
where $\rho$ is the energy density.

By Theorem \ref{Theo:Defo}, when the FLRW universe model is deformed by replacing
$\Sigma(\GNCI;\ConstCurvTwoDim=1)$ with $\gamma(\GNCI,\GNCII)$,
the Einstein equations is satisfied provided that a string cloud is introduced.
The Friedmann equation remains unchanged. 
Hence, although homogeneity and isotropy are broken, 
the cosmic evolution governed by the scale factor $a(t)$ is exactly the same as in the FLRW universe.
All effects of the deformation are compensated by the string cloud.

A notable feature of the string cloud on this cosmological spacetime is that
the Raychaudhuri equation for a null geodesic congruence directed along the string
is unaffected by the string cloud. 
Indeed, for such a geodesic congruence, both the shear and the twist vanish, and moreover 
the contribution from the string energy-momentum tensor is zero. 
Consequently, the Raychaudhuri equation coincides exactly with that of a FLRW universe.
This implies that as long as observations are made along the string direction,
cosmic expansion effects, such as a distance-redshift relation, are not influenced by the string
cloud. 
By contrast, observations made in other directions may exhibit deviations. 
Although the expansion of our string cloud universe is isotropic in the sense that it is governed
by a single scale factor, the distance-redshift relation nevertheless depends on the direction of
the observation.

\subsubsection{Deformations of LRS cosmological models}
We specialize to the case 
\begin{equation}
  a(t,w) = a(t), \quad b(t,w) = a(t)\WFLorentz(t) \eqcolon b(t). 
\end{equation}
The  rank-one deformed metric \eqref{eq:SCMetR1Uni} takes the form 
\begin{equation}
  \Met[a,b,\gamma]
  =
  -\dd{t}^2 + a^2(t) \dd{w}^2 + b^2(t) \qty(\dd{\chi}^2 + \gamma^2(\GNCI,\GNCII) \dd{\sigma}^2),
\end{equation}
This metric reduces to the Kantowski-Sachs model for $\gamma(\GNCI,\GNCII) =
\Sigma(\GNCI;\ConstCurvTwoDim=1)$,
and to the LRS Bianchi type I and III models for $\ConstCurvTwoDim = 0$ and $\ConstCurvTwoDim = -1$,
respectively.
We refer to these models as LRS cosmological models. 
Applying Theorem \ref{Theo:Defo} as in Sec.~\ref{subsec:DefoFLRW},
we obtain deformations of these LRS cosmological models. 
As in the deformations of the FLRW universe, 
the temporal evolution of the scale factors $a(t), b(t)$ remains unchanged under the deformation. 
    
\subsubsection{Deformations of topological Reissner-Nordström-(A)dS black holes}
\label{subsubsec:DefoRNBH}
We consider the static case by setting 
\begin{equation}
  a(t,w) = a(w), \quad b(t,w) = a(w) \WFRI(w) \eqcolon b(w), 
\end{equation}
and introduce a new coordinate $r \coloneq b(w)$. 
The deformed metric \eqref{eq:SCMetR1Uni} takes the form 
\begin{equation}
  \Met = - f_1(r) \dd{t}^2 + f_2(r) \dd{r}^2
  + r^2 \qty(\dd{\GNCI}^2+\gamma^2(\GNCI,\GNCII) \dd{\GNCII}^2). 
\end{equation}
This metric reduces to that of the topological Reissner-Nordström-(A)dS black hole
when 
\begin{equation}
  \gamma(\GNCI,\GNCII) = \Sigma(\GNCI;\ConstCurvTwoDim), \quad 
  f_1(r) = \frac{1}{f_2(r)} 
  = \ConstCurvTwoDim - \frac{2m}{r} + \frac{q^2}{r^2} + \frac{\Lambda}{3} r^2. 
\end{equation}

By Theorem \ref{Theo:Defo}, with $f_1(r)$ and $f_2(r)$ fixed,
the metric remains a solution of the Einstein equations when
$\Sigma(\GNCI;\ConstCurvTwoDim)$ is replaced by a generic function $\gamma(\GNCI,\GNCII)$,
provided that an appropriate string cloud is introduced.
The string cloud compensates for all effects of the deformation.
Thus, we obtain deformations of the topological Reissner-Nordström-(A)dS solution. 
In particular, for $\ConstCurvTwoDim = 1$, we recover the
black hole with stringy hair constructed by Boos and Frolov \cite{Boos2018BHString}. 

We now examine the Killing horizons associated with the static
Killing vector field $\pdv*{t}$, which are given by $f_1(r) = 0$.
Despite undergoing arbitrary deformations on the transverse two-dimensional geometry,
the null generators of the Killing horizon remain unchanged.
Indeed, the horizon generators coincide with the Killing vector field and form a congruence of null
geodesics with vanishing expansion and shear.
Moreover, since the string cloud is stretched perpendicularly to the two dimensions,
the string cloud contribution to the Raychaudhuri equation for the Killing horizon generators
vanishes. 
As a result, the Raychaudhuri equation remains unchanged under the deformation.
This shows that the Killing horizon is rigid under the present class of deformations.
The surface gravity is also unchanged.

\subsection{The rank-three case}

As noted in Sec.~\ref{subsec:SCMetric}, 
the rank-three deformed metric \eqref{eq:SCMetSPR3} describes
deformations of the LRS Bianchi models of type II, VIII, and IX.
We present two representative deformations: a deformation of the topological Einstein-Maxwell-Taub-NUT
solution and a deformation of the closed FLRW model. 

\subsubsection{Deformation of the topological Einstein-Maxwell-Taub-NUT-(A)dS solution}

For the rank-three metric \eqref{eq:SCMetSPR3},
we redefine the time coordinate $t$ and choose the scale factors $a(t)$ and $b(t)$ so that 
the metric takes the form 
\begin{equation}
  \Met
  =
  - \frac{\dd{t}^2}{U(t)}
  + (2l)^2 U(t) \qty(\dd{w} + f(\GNCI,\GNCII) \dd{\GNCII})^2
  + (t^2 + l^2) \qty(\dd{\GNCI}^2 + \gamma^2(\GNCI,\GNCII)  \dd{\GNCII}^2), 
\end{equation}
where $l$ is a constant, which will be identified with the NUT parameter below. 
This metric reduces to the topological Einstein-Maxwell-Taub-NUT-(A)dS solution
when
\begin{equation}
  \gamma(\GNCI,\GNCII) = \Sigma(\GNCI;\ConstCurvTwoDim),
  \quad
  U(t)
  =
  \frac{\ConstCurvTwoDim(t^2-l^2) -2 m t + q^2 - \Lambda(t^4/3 + 2l^2 t^2 - l^4)}
  {t^2 + l^2}. 
\end{equation}
Here $m$ denotes a mass parameter and $q$ an electric charge \cite{ExactSpacetimes}.
Applying Theorem \ref{Theo:Defo} with $U(t)$ fixed as in Sec.~\ref{subsubsec:DefoRNBH}
yields a class of deformations of the topological Einstein-Maxwell-Taub-NUT-(A)dS
solution.

\subsubsection{Deformation of the closed FLRW universe}

We specialize to the case of a single scale factor by setting $a(t) = b(t)$.
The metric \eqref{eq:SCMetSPR3} takes the form 
\begin{equation}
  \Met
  =
  -\dd{t}^2
  + a^2(t) \CompatMet, \quad
  \CompatMet =
  \frac{1}{4}
  \qty[
  (\dd{w} + f(\GNCI) \dd{\GNCII})^2
  +
  \dd{\GNCI}^2+\gamma^2(\GNCI,\GNCII)\dd{\GNCII}^2
  ],
\end{equation}
where we have redefined the scale factor by $a(t) \to a(t)/2$ 
in order that the three-metric $\CompatMet$ be normalized to the unit round $S^3$
when it takes the spherical form. 
Indeed, when $\gamma(\GNCI,\GNCII) = \Sigma(\GNCI;\ConstCurvTwoDim=1) = \sin \GNCI$, 
the spatial metric $\CompatMet$ reduces to that of the unit round $S^3$, 
written in coordinates adapted to the Hopf fibration. 
Consequently, the spacetime metric $\Met$ coincides with the closed FLRW metric.
Applying Theorem \ref{Theo:Defo}, we obtain deformations of the closed FLRW universe.
As in Sec.~\ref{subsec:DefoFLRW}, the temporal evolution of the scale factor is unchanged
by the deformation. 
The difference from the deformed FLRW universe discussed in Sec.~\ref{subsec:DefoFLRW} lies in
the string cloud.
In the present rank-three case, string worldsheets have nonvanishing twist potentials,
as in Ref.~\cite{Kozaki2024nullcontact}, 
whereas in the rank-one case the twist potentials vanish.

\section{Conclusion}
\label{sec:conclusion}

We established in Theorem~\ref{Theo:Defo} a systematic framework that realizes 
deformations of highly symmetric spacetimes as exact solutions to the Einstein equations
with a string cloud.
The framework is based on three-dimensional {\EtaEin} metrics compatible with
normal almost contact structures, 
from which the corresponding spacetime metrics are constructed.
The string cloud provides the energy-momentum tensor required to compensate for the
additional terms in the Einstein tensor induced by the deformation,
so that the Einstein equations are satisfied.

The spacetime metrics are given by Eqs.~\eqref{eq:SCMetR1Uni} and \eqref{eq:SCMetSPR3}.
The metric \eqref{eq:SCMetR1Uni} admits no Killing vector fields in general,
while the metric \eqref{eq:SCMetSPR3} admits only a single Killing vector field. 
Nevertheless, the resulting solutions arise as deformations of highly symmetric spacetimes. 
These include deformations of FLRW universes, the Kantowski-Sachs spacetime, 
locally rotationally symmetric (LRS) Bianchi type I, II, III, VIII, and IX spacetimes,
and Einstein-Maxwell-Taub-NUT-(A)dS solutions and Reissner-Nordström-(A)dS black holes,
each admitting spherical, planar, or hyperbolic symmetry.

For cosmological solutions, 
the evolution equations for the scale factors remain unchanged under the deformation. 
In this sense, the dynamical expansion history, as determined by the Einstein equations,
coincides with that of the corresponding symmetric solutions.
From an observational perspective, 
cosmological expansion is described through the behavior of null geodesic congruences. 
We examine the Raychaudhuri equation for a null geodesic congruence propagating 
along the string direction in the deformed FLRW universes.
We find that the Raychaudhuri equation also remains unchanged.
This implies that the presence of the string cloud leaves the distance-redshift relation
for null geodesics along the string direction unchanged. 
Null geodesic congruences in other directions have not been analyzed in the present work. 
It remains an open question whether the string cloud can induce an apparent acceleration
or deceleration through direction-dependent optical effects.

In the deformations of the Reissner-Nordström-(A)dS black holes admitting spherical, 
planar, or hyperbolic symmetry,
we found that the Killing horizons are unaffected by the deformation, 
except for possible changes in their intrinsic geometry.
This suggests that the spacetime structure,
at least at the level of horizon geometry,
is robust under the deformation. 
The deformed metrics of the Reissner-Nordström-(A)dS black holes are included in the
general family studied in Ref.~\cite{Koga2021PhotonSurf},
and metrics in this class admit null geodesics confined to hypersurfaces $r =
\text{const}.$ 
This indicates that photon surfaces exist even when black hole spacetimes are 
deformed. 
Possible distortions of the photon surface, if observable, could therefore provide a novel
probe of string cloud configurations.

\begin{acknowledgments}
  The authors thank Ken-ichi Nakao and Hirotaka Yoshino for useful discussions.
  HI and YM are partly supported by MEXT Promotion of Distinctive Joint
  Research Center Program JPMXP 0723833165. 
  TK is partly supported by JSPS KAKENHI Grant Number JP20K03772. 
\end{acknowledgments}

\appendix

\section{Local metric in a three-dimensional normal ACM manifold}
\label{app:MetInNormACM}
We show that a local metric in a three-dimensional normal ACM manifold is given by
Eq.~\eqref{eq:NormACMMet}.
The derivation is mainly based on the proof of Theorem 4 in \cite{Olszak1986NACM3D}. 
We denote the normal ACM structure as $\ACMStruct$. 

\newcommand{\QSMet}{\tilde{\CompatMet}}
First, we take a function $\Omega$ so that
\begin{equation}
  \label{eq:DefOmega}
  \XiVect \Omega = \frac{1}{2}\TrCovDerXi, 
\end{equation}
and construct a new metric 
\begin{equation}
  \label{eq:ConstQSMet}
  \QSMet = e^{-2\Omega} \CompatMet + (1 - e^{-2\Omega}) \EtaForm \otimes \EtaForm. 
\end{equation}
This metric is also compatible with the normal AC structure $\ACStruct$, 
and furthermore, the new ACM structure $(\PhiTens, \XiVect, \EtaForm, \QSMet)$ is quasi-Sasakian. 
Therefore, $\XiVect$ is a Killing vector field with respect to the $\QSMet$.

Next, we take local coordinates $w, x, y$ such that
\begin{equation}
  \label{eq:CondCoordW}
  \XiVect = \pdv{w}. 
\end{equation}
The coordinates are not unique.
Indeed, $w',x',y'$ defined by
\begin{equation}
  \label{eq:FreedomCoordTrans}
  w' \coloneq w + w'(x,y), \quad x' \coloneq x'(x,y), \quad y' \coloneq y'(x,y). 
\end{equation}
also satisfy the condition. 
Using this coordinate freedom, we can write the metric $\QSMet$ as 
\begin{equation}
  \label{eq:QSMetConfFlat}
  \QSMet = (\dd{w} + f(x,y) \dd{y})^2 + e^{2 \Sigma(x,y)}(\TwoDimFlatMet), 
\end{equation}
for some functions $f(x,y)$ and $\Sigma(x,y)$.
In these coordinates, it follows from Eq.~\eqref{eq:MetDualEtaXi} that the {\OneForm} $\EtaForm$ is
written as 
\begin{align}
  \EtaForm = \dd{w} + f(x,y) \dd{y}. 
\end{align}
While we have written the two-dimensional part in a conformally flat form,
it can also be expressed in Gaussian normal coordinates so that 
\begin{equation}
  \label{eq:QSMetGaussNom}
  \QSMet
  =
  (\dd{w} + f(\GNCI,\GNCII) \dd{\GNCII})^2 + \dd{\GNCI}^2
  + \gamma^2(\GNCI,\GNCII)\dd{\GNCII}^2. 
\end{equation}

Finally, from Eq.~\eqref{eq:ConstQSMet}, 
the original metric $\CompatMet$ can be expressed as 
\begin{equation}
  \CompatMet = (\dd{w} + f(x,y) \dd{y})^2 + e^{2 \Omega(w,x,y)} (\TwoDimFlatMet), 
\end{equation}
where we have absorbed $\Sigma(x,y)$ into the function $\Omega(w,x,y)$. 
This absorption does not affect Eq.~\eqref{eq:DefOmega},
and hence, for this metric representation,  it holds that 
\begin{equation}
  \TrCovDerXi = 2\Omega_{,w}. 
\end{equation}

\section{Solutions of Eqs.~\eqref{eq:CondEtaEinRank3A} and \eqref{eq:CondEtaEinRank3B}}
\label{sec:solutions-eqs}

In this appendix, we solve Eqs.~\eqref{eq:CondEtaEinRank3A} and \eqref{eq:CondEtaEinRank3B} under
a different assumption from the quasi-Sasakian condition adopted in Sec.~\ref{sec:almost-cont-struct}.
Specifically, we assume that
\begin{equation}
  \label{eq:Assumption}
 f_{,y} = \Omega_{,y} = 0.  
\end{equation}
Under this assumption, Eqs.~\eqref{eq:CondEtaEinRank3A} and \eqref{eq:CondEtaEinRank3B} reduce to 
\begin{align}
  \label{eq:NecSufCondEtaEinNonTriv1}
  \pdv{x} 2 \Omega_{,w} + f \pdv{w} \qty(e^{-2\Omega} f') &= 0,
  \\
  \label{eq:NecSufCondEtaEinNonTriv2}
  - f \pdv{w} 2 \Omega_{,w} + \pdv{x} \qty(e^{-2\Omega}f') &= 0, 
\end{align}
where $f'$ denotes the differentiation with respect to $x$. 
These equations are solved as follows. 

First, using Eqs.~\eqref{eq:NecSufCondEtaEinNonTriv2} and \eqref{eq:NecSufCondEtaEinNonTriv1}, 
we obtain 
\begin{equation}
   \pdv{w} \qty[2 \Omega_{,ww} e^{2 \Omega} - f'^2 e^{-2\Omega}] = 0, 
\end{equation}
which leads to 
\begin{equation}
  2 \Omega_{,ww} e^{2\Omega} - f'^2 e^{-2\Omega} = c_1, 
\end{equation}
where $c_1$ is a function of $x$. 
Next, multiplying this equation by $\Omega_{,w} e^{-2\Omega}$, we have
\begin{equation}
   \pdv{w}\qty[(\Omega_{,w})^2 + \frac{f'^2}{4} e^{-4\Omega} + \frac{c_1}{2} e^{-2\Omega}] = 0, 
\end{equation}
which leads to 
\begin{equation}
  (\Omega_{,w})^2  + \frac{f'^2}{4} e^{-2\Omega} + \frac{c_1}{2} e^{-2\Omega} = \frac{c_2}{4},
\end{equation}
where $c_2$ is a function of $x$. 
To integrate this equation, we multiply both sides by $(2e^{2\Omega})^2$.
Then we have
\begin{equation}
\label{eq:FinEq}
   \qty(\pdv{J}{w})^2 = c_2 J^2 - 2c_1 J - f'^2, 
\end{equation}
where $J\coloneq e^{2 \Omega}$.
This equation is solved explicitly as 
\begin{equation}
   \label{eq:SolutionOfJ}
   J =
   \begin{dcases}
     - \dfrac{f'^2}{2c_1} - \dfrac{c_1}{2} (w - w_0)^2, 
     & c_2 = 0, \\
     \dfrac{c_1}{c_2}
     +
     \sqrt{\qty(\frac{c_1}{c_2})^2 + \frac{f'{}^2}{c_2}} \cosh \sqrt{c_2}(w - w_0), 
     & c_2 > 0, \\
     \dfrac{c_1}{c_2}
     +
     \sqrt{\qty(\frac{c_1}{c_2})^2 + \frac{f'{}^2}{c_2}} \cos \sqrt{-c_2}(w - w_0), 
     & c_2 < 0, 
   \end{dcases}
\end{equation}
where $w_0$ is a function of $x$.

The functions $c_1, c_2$, and $w_0$ must be chosen so that
Eqs.~\eqref{eq:NecSufCondEtaEinNonTriv1} and \eqref{eq:NecSufCondEtaEinNonTriv2} are satisfied.
We analyze these constraints in detail and explicitly construct the corresponding {\EtaEin}-metrics. 

\subsection{The case that $c_2=0$}
\label{sec:case-that-c_2=0}

Substituting Eq.~\eqref{eq:SolutionOfJ} to Eqs.~\eqref{eq:NecSufCondEtaEinNonTriv1} and
\eqref{eq:NecSufCondEtaEinNonTriv2},
we find that $w_0$ is a constant, which can be set to zero, and $c_1$ is given
by 
\begin{equation}
  c_1 = \dfrac{F''}{F}, 
\end{equation}
where
\begin{equation}
  F \coloneq \int f \, \dd{x}. 
\end{equation}
In this setting, the metric is written as 
\begin{equation}
  \label{eq:Rank3NACMEtaEinRiemMet}
   \CompatMet = (\dd{w} + F' \dd{y})^2 + \frac{w^2 + F^2}{2} \TwoMet, 
\end{equation}
where $\TwoMet$ is the two-dimensional metric defined by 
\begin{equation}
  \TwoMet = -\frac{F''}{F} \qty(\dd{x}^2 + \dd{y}^2). 
\end{equation}
This is the explicit {\EtaEin} metric under the assumption \eqref{eq:Assumption} with $c_2 = 0$. 
The Ricci tensor is calculated as
\begin{equation}
   \label{eq:Lambda}
   \RicTens = \lambda \CompatMet - \lambda \EtaForm \otimes \EtaForm,
   \quad
   \lambda = \frac{R_2 - 1}{w^2 + F^2}. 
\end{equation}
where $\TwoRic$ is the Ricci scalar of the two-dimensional metric $\TwoMet$, 
which is simply given by
\begin{equation}
  \label{eq:TwoRicScaVaniC2}
  \TwoRic = -\frac{F}{F''} \qty[\ln (-\frac{F}{F''})]''. 
\end{equation}

\subsection{The case that $c_2 > 0$}
\newcommand{\am}{\text{am}}
\newcommand{\dn}{\text{dn}}
Substituting Eq.~\eqref{eq:SolutionOfJ} to Eqs.~\eqref{eq:NecSufCondEtaEinNonTriv1} and
\eqref{eq:NecSufCondEtaEinNonTriv2},
we find that the functions $c_2$ and $w_0$ are constants and $c_1$ is given by 
\begin{equation}
  c_1 = - \sqrt{c_2} F'' \tan \sqrt{c_2} F. 
\end{equation}
This gives the explicit {\EtaEin} metrics under the assumption \eqref{eq:Assumption}
with a positive $c_2$.

A simple example is given by taking 
\begin{equation}
  w_0 = 0, \quad c_2 = 4, \quad F = \am(x,2) + \frac{\pi}{4}, 
\end{equation}
where $\am(x,2)$ denotes the Jacobi amplitude with parameter $2$.
In this setting, the metric function $f$ is given by
\begin{equation}
  f = \dn(x,2), 
\end{equation}
where $\dn(x,2)$ is the Jacobi delta amplitude. 
The metric is written as 
\begin{equation}
  \CompatMet = (\dd{w} + f \dd{y})^2 +
  \frac{\cosh 2w - f^2}{2} (\TwoDimFlatMet), 
\end{equation}
and the Ricci tensor takes the form of {\EtaEin} 
\begin{equation}
  \RicTens
  =
  - \frac{2\cosh 2w}{\cosh 2w - f^2} \CompatMet
  +
  \frac{2f^2}{\cosh 2w - f^2} \EtaForm\otimes \EtaForm. 
\end{equation}

\subsection{The case that $c_2 < 0$}
Substituting Eq.~\eqref{eq:SolutionOfJ} to Eqs.~\eqref{eq:NecSufCondEtaEinNonTriv1} and
\eqref{eq:NecSufCondEtaEinNonTriv2},
we find that the functions $c_2$ and $w_0$ are constants and $c_1$ is given by 
\begin{equation}
  c_1 = \sqrt{- c_2} F'' \coth \sqrt{-c_2} F. 
\end{equation}
This gives the explicit {\EtaEin} metrics under the assumption \eqref{eq:Assumption}
with a negative $c_2$.

\section{Derivation of Eq.~\eqref{eq:EneMomSingleString}}
\label{appendix:deriv}
We first review the general framework of a string energy-momentum tensor. 
Let $\zeta^A~(A=0,1)$ and $x^{\mu}~(\mu=0,\dots,3)$ be coordinates on a string worldsheet and a
spacetime manifold.
Let $x^{\mu} = X^{\mu}(\zeta^a)$ be the embedding of the string worldsheet.
Then the string action $S_{\text{string}}$ is given by
\begin{align}
  S_{\text{string}}[g,X] = -\mu \int \sqrt{-\gamma} \, \dd^2\zeta = \int \mathcal{L}_{\text{string}}
  \sqrt{-g} \, \dd^4 x, \\
  \mathcal{L}_{\text{string}} \coloneq - \frac{\mu}{\sqrt{-g}} \int \delta^4(x-X(\zeta))
  \sqrt{-\gamma} \, \dd^2 \zeta, 
\end{align}
where $\gamma$ is the determinant of the worldsheet metric
\begin{equation}
  \gamma_{AB} = g_{\mu\nu} \frac{\partial X^{\mu}}{\partial \zeta^A} \frac{\partial
    X^{\nu}}{\partial  \zeta^{B}}. 
\end{equation}
The energy-momentum tensor of the string is given by
\begin{equation}
  T^{\mu\nu}_{\text{string}}
  =
  \frac{2}{\sqrt{-g}} \fdv{S_{\text{string}}}{g_{\mu\nu}}
  =
  -\frac{\mu}{\sqrt{-g}} \int \delta^4(x-X(\zeta)) \sqrt{-\gamma} \, \gamma^{AB}
  \pdv{X^{\mu}}{\zeta^{A}} \pdv{X^{\nu}}{\zeta^{B}} \dd^2 \zeta. 
\end{equation}

We then apply the framework to the current setting.
For the $\xi$-string specified by $(\GNCI,\GNCII) = (\GNCI_0, \GNCII_0)$,
taking the following embedding:
\begin{equation}
  t(\zeta^A) = \zeta^0, \quad w(\zeta^A) = \zeta^1,
  \quad \GNCI(\zeta^A) =  \GNCI_0, \quad \GNCII(\zeta^A) = \GNCII_0, 
\end{equation}
we have Eq.~\eqref{eq:EneMomSingleString}.

\bibliography{library}
\end{document}